\documentclass[a4paper,11pt]{article}
\usepackage{epsfig}
\usepackage{verbatim}
\usepackage{amsmath} % font ams
\usepackage{amssymb}
\usepackage[algoruled,vlined,english]{algorithm2e} 
\usepackage{enumerate}
\usepackage{xspace}
\usepackage{multirow}
\usepackage{authblk}
\usepackage{graphicx}
\newcommand{\code}[1]{\texttt{#1}}
\newcommand{\Pairedwalk}{{\em Paired Walk}~}
 
\SetAlFnt{\footnotesize}

 \newcommand{\qed}{\hspace*{\fill}\rule{6pt}{6pt}\vspace{.5\smallskipamount}}

 \newtheorem{theorem}{Theorem}[section]

 \newtheorem{lemma}{Lemma}[section]

 \newenvironment{proof} { \noindent \emph{Proof} : } { \qed }

\newcommand{\cS}{{\cal S}}
\newcommand{\A}{{\cal A}}
\newcommand{\cC}{{\cal C}}
\newcommand{\cD}{{\cal D}}

\begin{document}

\title{Tight Bounds for Black Hole Search with Scattered Agents in Synchronous Rings\thanks{Part of this work was done while E. Markou was visiting the LIF research laboratory in Marseille, France. Authors J.~Chalopin, S.~Das and A.~Labourel are partially supported by ANR projects SHAMAN and ECSPER.}}
%{Locating Black Holes with Scattered Mobile Agents in Synchronous Rings}
%{Black Hole Search with Dispersed Finite State Agents in a Ring}
%{All you need is a token or two: Black hole Search with Scattered Agents}
%{How many tokens do you need? Scattered Black Hole Search in Rings}
%{Meet me at the black gates: Scattered agents can find a black hole}

\author[1]{J\'{e}r\'{e}mie Chalopin} 
\author[1]{Shantanu Das}
\author[1]{Arnaud Labourel}
\author[2]{Euripides Markou}

\affil[1]{LIF, CNRS \& Aix Marseille University, 39 rue Joliot Curie, 13453 Marseille,
France. {\small Email:~\{jeremie.chalopin,shantanu.das,arnaud.labourel\}@lif.univ-mrs.fr}}

\affil[2]{Department of Computer Science and Biomedical Informatics,
University of Central Greece, Lamia, Greece. {\small Email:~emarkou@ucg.gr}}

%\author{J\'{e}r\'{e}mie Chalopin\inst{1} \and Shantanu Das\inst{1} \and Arnaud Labourel\inst{1} \and Euripides Markou\inst{2}}

%\institute{LIF, CNRS \& Aix-Marseille University, Marseille, France. \email{~\{jeremie.chalopin,shantanu.das,arnaud.labourel\}@lif.univ-mrs.fr} 
%%39 rue Joliot Curie, 13453 Marseille,  
%\and {Department of Computer Science and Biomedical Informatics, University of Central Greece, Lamia, Greece. \email{~emarkou@ucg.gr}}}

% Footnotes to be added for final version
%\footnotetext[1]{ \noindent 
%Part of this work was done while the fourth author was visiting the LIF research laboratory in Marseille, France.}
%\footnotetext[2]{ \noindent 
%These authors were partially supported by ANR project SHAMAN}

%\pagestyle{empty}
%\pagenumbering{}
\date{}
\maketitle

\begin{abstract}
We study the problem of locating a particularly dangerous node, the so-called \emph{black hole} in a synchronous anonymous ring network with mobile agents. A black hole is a harmful stationary process residing in a node of the network and destroying
destroys all mobile agents visiting that node without leaving any trace. Unlike most previous research on the {\em black hole search} problem which employed a colocated team of agents, %has been solved previously with teams of co-located agents having distinct identities. 
we consider the more challenging scenario when the agents are identical and initially scattered within the network. % and thus, it is difficult to achieve coordination as the agents have no means of communicating from a distance and have no prior knowledge of the number and locations of the agents. 
Moreover, we solve the problem with agents that have constant-sized memory and carry a constant number of identical tokens, which can be placed at nodes of the network. 
In contrast, the only known solutions for the case of scattered agents searching for a black hole, use stronger models where the agents have non-constant memory, can write messages in whiteboards located at nodes or are allowed to mark both the edges and nodes of the network with tokens. 
This paper solves the problem for ring networks containing a single black hole.  
We are interested in the minimum resources (number of agents and tokens) necessary for locating all links incident to the black hole. We present deterministic algorithms for ring topologies and 
provide matching lower and upper bounds for the number of agents and the number of tokens required for deterministic solutions to the black hole search problem, in oriented or unoriented rings, using movable or unmovable tokens. 
%in this scenario.
%\end{abstract}

\bigskip

%\begin{key}
\noindent {\textbf{Keywords:}} {Distributed Algorithms, Fault Tolerance, Black Hole Search, Mobile Agents, Anonymous Networks, Identical tokens, Finite State Automata}
%\end{key}
\end{abstract}
%\newpage
%\pagestyle{plain}
%\pagenumbering{arabic}

\section{Introduction}

\subsection{Overview}

We consider the problem of exploration in unsafe networks which contain malicious hosts of a highly harmful nature, called {\em black holes}. A black hole is a node which contains a stationary process destroying all mobile agents visiting this node, without leaving any trace~\cite{dfps07}. 
In the {\em Black Hole Search} (BHS) problem the goal for a team of agents is to locate the black hole within finite time, with the additional constraint that at least one of the agents must remain alive. 
In particular, at least one agent must survive and the surviving agents must have located (or marked) all edges leading to the black hole. 
It is usually assumed that all the agents start from the same location and have distinct identities. In this paper, we do not make such an assumption and study the problem for identical agents starting from distinct locations within the network. We focus on minimizing the resources required to find the black hole. 

The only way of locating a black hole is to have at least one agent visiting it. However, since any agent visiting a black hole is destroyed without leaving any trace, the location of the black hole must be deduced by some communication mechanism employed by the agents. Four such mechanisms have been proposed in the literature: a) the {\em whiteboard} model in which there is a whiteboard at each node of the network where the agents can leave messages, b) the {\em `pure' token} model where the agents carry tokens which they can leave at nodes,  c) the {\em `enhanced' token} model in which the agents can leave tokens at nodes or edges, and d) the time-out mechanism (only for synchronous networks) in which one agent explores a new node while another waits for it at a safe node. 

The most powerful inter-agent communication mechanism is having whiteboards at all nodes. Since access to a whiteboard is provided in mutual exclusion, this model could also provide the agents a symmetry-breaking mechanism: If the agents start at the same node, they can get distinct identities and then the distinct agents can assign different labels to all nodes. Hence in this model, if the agents are initially co-located, both the agents and the nodes can be assumed to be non-anonymous without any loss of generality. 
The BHS problem has been studied using whiteboards in asynchronous networks, with the objective of minimizing the number of agents required to locate the black hole. Note that in asynchronous networks, it is not possible to answer the question of whether or not a black hole exists in the network, since there is no bound on the time taken by an agent to traverse an edge.
Assuming the existence of (exactly one) black hole, the minimum sized team of co-located agents that can locate the black hole depends on the maximum degree $\Delta$ of a node in the network (unless the agents have a complete map of the network). In any case, the prior knowledge of the network size is essential to locate the black hole in finite time.

In the case of synchronous networks two co-located distinct agents can discover one black hole in any graph by using the time-out mechanism, without the need of whiteboards or tokens. 
Furthermore it is possible to detect whether a black hole actually exists or not in the network. Hence, with co-located distinct agents, the issue is not the feasibility but the time efficiency of black hole search (see \cite{ckr06,ckr10,ckmp06,ckmp07,kmrs07,kmrs08,knp09} for example). However when the agents are scattered in the network (as in our case), the time-out mechanism is not sufficient to solve the problem anymore.

Most of the previous results on black hole search used agents whose memory is at least logarithmic in the size of the network. This means that these algorithms are not scalable to networks of arbitrary size. This paper considers agents modeled as finite automata, i.e., having a constant number of states. This means that these agents can not remember or count the nodes of the network that they have explored. In this model, the agents cannot have prior knowledge of the size of the network. 
%It's not possible to detect the existence of a black hole in finite time. 
For synchronous ring networks of arbitrary size, containing exactly one black hole, we present deterministic algorithms for locating the black hole using scattered agents each having constant-sized memory.
We are interested in minimizing both the number of agents and the number of tokens required for solving the BHS problem.

We use the `pure' token model. While the whiteboard model is commonly used in unsafe networks, the token model has been mostly used for exploration of safe networks. 
Note that the `pure' token model can be implemented with $O(1)$-bit whiteboards (assuming that only a constant number of tokens may be placed on a node at the same time), while the `enhanced' token model can be implemented with $O(\log \Delta)$-bit whiteboards. In the previous results using the whiteboard model, the capacity of each whiteboard is always assumed to be of at least $\Omega (\log n)$ bits, where $n$ is the number of nodes of the network. Unlike the whiteboard model, we do not require any mutual exclusion mechanism at the nodes of the network.
We distinguish movable tokens (which can be picked up from a node and placed on another) from unmovable tokens (which can not be picked up once they are placed on a node). For both types of tokens, we provide matching upper and lower bounds on both the number of agents and the number of tokens per agent, required for solving the black hole search problem in synchronous rings. 
Although our algorithms require only a constant size memory for each agent, the impossibility results presented in this paper hold even for agents having unbounded memory.

%In this paper we study the Black Hole Search problem for scattered identical anonymous agents with constant memory in synchronous anonymous networks under the `pure token' model. We focus on the ring topology. We present deterministic algorithms that solve the problem in those topologies using a minimal number of agents and tokens. These are the first results for the Black Hole Search problem with scattered agents under the weakest `pure token' model.

%-------------------------------------
\subsection{Related Works}

The exploration of an unknown graph by one or more mobile agents is a classical problem initially formulated in 1951 by Shannon \cite{sha51} and it has been extensively studied since then (e.g., see \cite{bs94, dp99, fgkp06}). 
In unsafe networks containing a single dangerous node (black hole), the problem of searching for it has been studied in the asynchronous model using whiteboards and given that all agents initially start at the same safe node (e.g.,~\cite{dfkprs06,dfps06,dfps07,dfs04}). It has also been studied using `enhanced' tokens in \cite{dfks06,dkss06,shi09} and in the `pure' token model in \cite{fis08}.  It has been proved that the problem can be solved with a minimal number of agents performing a polynomial number of moves. Notice that in an asynchronous network the number of the nodes of the network must be known to the agents otherwise the problem is unsolvable \cite{dfps07}. If the network topology is unknown, at least $\Delta +1$ agents are needed, where $\Delta$ is the maximum node degree in the graph \cite{dfps06}. It is usually assumed that the network is bi-connected and the existence of one black hole is common knowledge.

In asynchronous networks, with scattered agents (not initially located at the same node), the problem has been investigated for arbitrary topologies \cite{cds07,fkms09} in the whiteboard model while in the `enhanced' token model it has been studied for rings \cite{dss07,dss08} and for some interconnected networks \cite{shi09}.  

The issue of efficient black hole search has been studied in synchronous networks without whiteboards or tokens (only using the time-out mechanism) in \cite{ckr06,ckr10,ckmp06,ckmp07,kmrs07,kmrs08,knp09} under the condition that all distinct agents start at the same node. 

The problem has also been studied for co-located agents in asynchronous and synchronous directed graphs with whiteboards in \cite{cdkmp09,knp09}. In \cite{ckr10} they study how to locate and repair faults (weaker than black holes) using co-located agents in synchronous known networks with whiteboards and in \cite{gla09} they study the problem in asynchronous networks with whiteboards and co-located agents without the knowledge of incoming link. A different dangerous behavior is studied for co-located agents in \cite{km10}, where the authors consider a ring and assume black holes with Byzantine behavior, which do not always destroy a visiting agent.

In all previous papers (apart from \cite{fis08}) studying the Black Hole Search problem using tokens, the `enhanced' token model is used. The weakest `pure' token model has only been used in \cite{fis08} for co-located agents in asynchronous networks. In all previous solutions to the problem using tokens, the agents are assumed to have non-constant memory.

%---------------------------------------
\subsection{Our Contributions}

%In this paper, we are interested in detection of a dangerous node called \emph{black hole} by a team of mobile agents. While previous studies on the black-hole search problem were mostly restricted to teams of co-located agents with distinct identities, we consider the more difficult scenario of anonymous identical agents that are initially scattered in an anonymous synchronous ring. In this scenario, coordinating the agents is difficult since the considered agents are anonymous and have no means of communicating from a distance. 
%\bigskip
%\noindent \textbf{Our Contributions:} 

Unlike previous studies on BHS, we consider the scenario of anonymous (i.e., identical) agents that are initially scattered in an anonymous ring.  
We focus our attention on very simple mobile agents. The agents have constant-size memory, they carry a constant number of identical tokens which can be placed at nodes and, (apart from using the tokens), they can communicate with other agents only when they meet at the same node. We consider four different scenarios depending on whether the tokens are movable or not, and whether the agents agree on a common orientation.  We present deterministic optimal algorithms and provide matching upper and lower bounds for the number of agents and the number of tokens required for solving BHS (See Table~\ref{tab:res} for a summary of results). Surprisingly, the agreement on the ring orientation does not influence the number of agents needed in the case of movable tokens but is important in the case of unmovable tokens. 

The lower bounds presented in this paper are very strong in the sense that they do not allow any trade-off between the number of agents and the number of tokens for solving the BHS problem. In particular we show that: \begin{itemize}
\item Any constant number of agents, even having unlimited memory, cannot solve the BHS problem with less tokens than depicted in all cases of Table~\ref{tab:res}.

\item Any number of agents less than that depicted in all cases of Table~\ref{tab:res} cannot solve the BHS problem even if the agents are equipped with any constant number of tokens and they have unlimited memory.
\end{itemize}
Meanwhile our algorithms match the lower bounds, are time-optimal and since they do not require any knowledge of the size of the ring or the number of agents, they work in any anonymous synchronous ring, for any number of anonymous identical agents (respecting the minimal requirements of Table~\ref{tab:res}). 
%Due to space limitations, proofs and formal algorithms are omitted and can be found in the full version of the paper~\cite{CDLM2011}.

%The main question answered by the paper is how the limitation on the memory of the agents influences the resources required for solving BHS. We show that the constant memory limitation has no influence on the resource requirements since the (matching) lower bounds hold even if the agents have unlimited memory.

\begin{table}[ht]
\begin{center}
\begin{tabular}{| c | c | c | c | c |}\cline{3-4}
\multicolumn{2}{c}{}&\multicolumn{2}{|c|}{Resources necessary } \\ 
\multicolumn{2}{c}{}&\multicolumn{2}{|c|}{ and sufficient }  \\ 
\cline{1-5}
\textbf{Tokens are} & \textbf{Ring is} & \textbf{\# agents} &  \textbf{\# tokens} & \textbf{References in the paper}\\  
\hline
\multirow{2}{*}{Movable }&Oriented &\multirow{2}{*}{3} &\multirow{2}{*}{1}&\multirow{2}{*}{Theorem~\ref{impos-k-1}, \ref{th:impo-2oriented} and~\ref{lem:BHS-R-3-1m}}\\ \cline{2-2}
&Unoriented &&&\\ \hline
\multirow{2}{*}{Unmovable}&Oriented &4 &2&Theorem~\ref{impos-k-1}, \ref{th:impo-3oriented} and \ref{th:BHS-R-4-2u}\\  \cline{2-5}
&Unoriented &5 &2&Theorem~\ref{impos-k-1}, \ref{th:impo-4unroriented} and~\ref{th:BHS-R-5-2u}\\ \hline
\end{tabular}
\end{center}
\caption{Summary of results for BHS in synchronous rings \label{tab:res}}
\end{table}

%----------------------------
\section{Our Model}
%---------------------------

Our model consists of an anonymous, synchronous ring network with ${k \geq 2}$ identical mobile agents that are initially located at distinct nodes called \emph{homebases}. Each mobile agent owns a constant number $t$ of identical tokens which can be placed at any node visited by the agent. The tokens are indistinguishable. Any token or agent at a given node is visible to all agents on the same node, but not visible to agents on other nodes. The agents follow the same deterministic algorithm and begin execution at the same time and being in the same initial state. In all our protocols a node may contain at most two tokens at the same time.
At any node of the ring, the ports leading to the two incident edges are distinguishable and locally labelled (e.g. as $1$ and $2$) and an agent arriving at a node knows the port-label of the edge through which it arrived. In the special case of an oriented ring, the ports are consistently labelled as {\tt Left} and {\tt Right} (i.e., all ports going in the clockwise direction are labelled {\tt Left}). In an unoriented ring, the local port-labeling at a node is arbitrary and each agent in its first step chooses one direction as {\tt Left} and in every subsequent step, it translates the local port-labeling at a node into {\tt Left} and {\tt Right} according to its chosen orientation.  
 
In a single time unit, each mobile agent completes one \emph{step} which consists of the \emph{Look}, \emph{Compute} and \emph{Move} stages (in this order). During the Look stage, an agent obtains information about the configuration %$c_v \in \cC_v$ 
of the current node (i.e., agents, tokens present at the node) and its own configuration %$c_a \in \cC_A$ 
(i.e., the port through which it arrived and the number of tokens it carries). During the Compute stage, an agent can perform any number of computations (i.e., computations are instantaneous in our model). During the Move stage, the agent may put or pick up a token at the current node and then either move to an adjacent node or remain at the current node. Since the agents are synchronous they perform each stage of each step at the same time. 
We call a token {\em movable} if it can be put on a node and picked up later by any mobile agent visiting the node. Otherwise we call the token {\em unmovable} in the sense that, once released, it can occupy only the node where it was released. 

Formally we consider a mobile agent as a finite Moore automaton 
$\A=(\cS,S_0,\Sigma,\Lambda,\delta,\phi)$, where $\cS$ is a set of $\sigma \geq 2$ states among which there is a specified state $S_0$ called the {\em initial} state; $\Sigma\subseteq\cD\times \cC_v\times\cC_{A}$ is the set of possible configurations an agent can see when it enters a node; $\Lambda\subseteq \cD\times\{\tt put, pick, no\mbox{ }action\}$ is the set of possible actions by the agent; $\delta: \cS\times \Sigma\rightarrow\cS$ is the transition function; and $\phi:\cS\rightarrow\Lambda$ is the output function. $\cD = \{\tt left, right, none\}$ is the set of possible directions 
through which the agent arrives at or leaves a node ($\tt none$ represents no move by the agent). $\cC_v= {\{ 0,1\}^{\sigma}} \times \{ 0,1,2 \}$ is the set of possible configurations at a node, consisting of a bit string that denotes for each possible state whether there is an agent in that state, and an integer that denotes the number of tokens at that node (in our protocols at most $2$ tokens reside at a node at any time). Finally, $\cC_{A}=\{ 1, 2\} \times \{ 0, 1 \}$ is the set of possible configurations of an agent, i.e., its orientation and whether it carries any tokens or not.

%Initially an agent is at some node $u_0$ in the initial state $S_0\in\cS$. $S_0$ determines an action (put token or nothing) and a direction from which the agent leaves $u_0$, $\phi(S_0)\in \Lambda$ . When incoming to a node $v$, the behavior of the agent is as follows. It reads the incoming port $i$ ($none$ if it did not move in the previous step), then
%the configuration $c_v \in \cC_v$ of node $v$ (i.e., the number of tokens and agents in $v$ and the states of the other agents) and of course the configuration $c_{MA} \in \cC_{MA}$ of the agent itself (i.e., the number of tokens carried). The triple $(i,c_v, c_{MA})\in \Sigma$ is an 
%input symbol that causes the transition from state $S$ to state $S'=\delta(S,(i,c_v,c_{MA}))$. $S'$ determines an action (such as put or pick a token or nothing) and a port direction $\phi(S')$, from which the agent leaves $v$.

%We assume that the memory of an agent is proportional to the number of bits required to encode its states which we take to be $\Theta( \log ( | \cS | ) )$ bits. 
Notice that all computations by the agents are independent of the size $n$ of the network and the number $k$ of agents. The agents have no knowledge of $n$ or $k$. 
The agents only know the number of tokens they have. Since the agents are identical they face the same limitations on their knowledge of the network. %In what follows, we assume that the agents have no knowledge about the number of nodes or any other parameter of the network. The agents start at the same initial state $S_0$ and at the same time. 
There is exactly one black hole in the network. An agent can start from any node other than the black hole and no two agents are initially colocated\footnote{Since there is no symmetry breaking mechanism at a node, two agents starting at the same node and in the same state, would behave as a single (merged) agent.}.  Once an agent detects a link to the black hole, it marks the link permanently as dangerous (i.e., disables this link). %Since the agents do not have enough memory to remember the location of the black hole, 
We require that at the end of a black hole search scheme, all links incident to the black hole (and only those links) are marked dangerous and that there is at least one surviving agent. Note that our definition of a successful BHS scheme is slightly different from the original definition, since we consider finite state agents.
%
% A black hole search scheme ({\em BHS-scheme}) should have the following property: Upon completion of the scheme there is at least one surviving agent, and this agent knows the location of the black hole.  The objective is to mark all the links incident to the black hole. The agents do not have any initial knowledge about the size of the network and are not designed for a specific size of ring.
The time complexity of a BHS scheme is the number of time units needed for completion of the scheme, assuming the worst-case location of the black hole and the worst-case initial placement of the scattered agents.

%We will focus on oriented and unoriented rings. The input in any of our algorithms is the number of agents (which are anonymous, identical, and start at the same initial state) and the number of tokens (which are indistinguishable and may be movable or unmovable) which represents an anonymous, synchronous and sometimes oriented ring network.

%----------------------------
\section{Impossibility Results}
%---------------------------

\subsection{Oriented Rings}

We first show that when the tokens are unmovable, any constant number of agents need 
at least two tokens to solve the BHS problem.

\begin{theorem}\label{impos-k-1}
For any constant $k$, there exists no algorithm that solves BHS in all oriented rings containing one black hole and $k$ or more scattered agents, when each agent is provided with only one unmovable token. The result holds even if the agents have unlimited memory.
\end{theorem}%For any constant $k>0$, there exists no algorithm that $k$ or more mobile agents carrying one unmovable token each, cannot solve the BHS problem in all oriented rings, even if the agents have unlimited memory.

\begin{proof}
Suppose there is a correct BHS algorithm that solves the problem with $k$ or more agents in rings of any size. If the algorithm does not require any agent to put down its token, such an algorithm should work even if every agent puts its token on its homebase in the first step. So without loss of generality, we assume that an agent puts down its token after executing a finite number of steps of the algorithm (unless it encounters some agents, some tokens or the black hole within this time). Now consider the behavior of this agent when placed on an infinite line (with no other agent). Suppose the agent puts down its token at a distance of $x$ (w.l.o.g. in the left direction) from its homebase. Further let $p$ be the maximum distance that the agent has travelled from its homebase (in either direction) before it puts down its token (thus, $x \leq p$). Now, consider a ring $R_1$ of size $n=2k(p+1)$ with one black hole and $k$ agents such that the agents are initially placed at a distance of $2(p+1)$ apart (see Figure~\ref{fig-impos-k-1}(a)). The black hole is located in the middle of a segment between two consecutive agents (i.e., it is at a distance $(p+1)$ from the closest agents). Since the agents start in the same state, they take the same actions and remain in identical states (until they encounter another agent, or a token or the black hole). As long as the agents do not travel any further than a distance $p$ from their homebase, they will not see anything different from the agent on the infinite line. Thus, each agent will put its token $x\leq p$ places to the left of its homebase. Each agent will do so at the same time and in the same state. During the rest of the algorithm, an agent can only move, observe the tokens and possibly mark some link as dangerous. 
Due to the correctness of the algorithm, at some time $\tau$, some agent will mark one link leading to the black hole as dangerous. Up to time $\tau$, all surviving agents behave the same. If there are more than one surviving agent, all of them are in the same state $\beta$ and still at a distance of $2(p+1)$ apart. Thus, if one such agent marks a link as dangerous, the next agent would mark a link at a distance of $2(p+1)$ away. So, one of the agents would have incorrectly marked a link---a contradiction to the correctness of the algorithm.

The remaining case to consider is when the agent that marks a link, is the last surviving agent. In this case, we can construct another ring $R_2$ of size $n=2(k+1)(p+1)$ with one black hole and $(k+1)$ agents initially placed the same distance apart as in the ring $R_1$ (see Figure~\ref{fig-impos-k-1}(b)). In an execution of the same algorithm on ring $R_2$, after exactly $\tau$ time steps, there will be two surviving agents both in the same state $\beta$ and at a distance of $2(p+1)$ from each other. Thus the two agents will mark as dangerous, two distinct links at a distance of $2(p+1)$ apart. Hence the algorithm fails for the ring $R_2$.  
\end{proof}

\begin{figure}[t]
\begin{center}
\scalebox{0.8}{\includegraphics{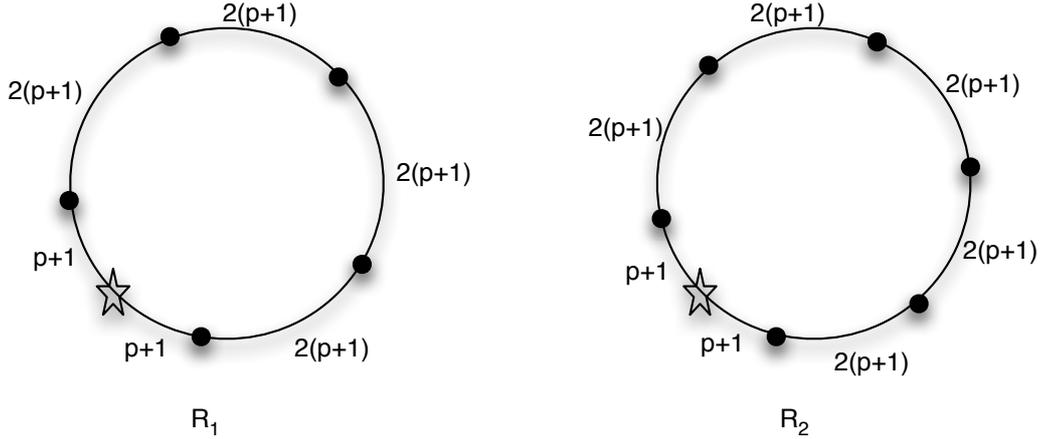}}
\end{center}
\vspace{-1cm}
\caption{\label{fig-impos-k-1}
(a) An oriented ring $R_1$ with $k$ agents and a black hole, (b) a larger oriented ring $R_2$ with $k+1$ agents and one black hole.}
\end{figure}

\medskip

We now derive some lower bounds on the number of agents necessary to solve the BHS problem.
The following result proves that at least one agent needs to be sacrificed for detecting each link leading to the black hole.

\begin{lemma}\label{lem:1}
During any execution of any BHS algorithm, if a link to the black hole is correctly marked, then at least one agent must have entered the black hole through this link.
\end{lemma}

\begin{proof}
Suppose for the sake of contradiction that there exists a correct BHS algorithm such that during any execution of this algorithm one link incident to the black hole is marked before any agent traverses this link. Assume without loss of generality that the link is on the left of the black hole. Consider the ring of size $n$ which leads to this execution. Now, add a vertex on the left of the black hole, obtaining a ring of size $n+1$, while keeping the same initial positions for the agents. The agents will behave as in the ring of size $n$ since they do not know the size of the ring and will see exactly the same configuration. Hence they will mark the left link of the new node as the link leading to the black hole. This contradicts with the correctness of the BHS algorithm.
\end{proof}

To solve the BHS problem in a ring, both links leading to the black hole need to be marked as dangerous. Thus, we immediately arrive at the following result.

\begin{theorem}\label{th:impo-2oriented}
Two mobile agents carrying any number of movable (or unmovable) tokens each, cannot solve the BHS problem in an oriented ring, even if the agents have unlimited memory.
\end{theorem}

%\begin{proof}
%The black hole cannot be located by two agents since, by Lemma~\ref{lem:1}, both agents must die before the black hole is marked from both directions.  
%\end{proof}

\noindent When the tokens are unmovable, even three agents are not sufficient to solve BHS as shown below. 

\begin{theorem}\label{th:impo-3oriented}
Three mobile agents carrying a constant number of unmovable tokens each, cannot solve the BHS problem in an oriented ring, even if agents have unlimited memory.
\end{theorem}

\begin{proof}
Suppose for the sake of contradiction that there exists an algorithm which solves the BHS problem for three agents each carrying a constant number $t$ of unmovable tokens. Let $x$ and $y$ be two integers chosen by the adversary, such that $1\leq x,y\leq 2t$. Suppose the three agents are initially placed on a ring of size $8t+x+y$ such that the distance between the first and second and  between the second and third agent is $4t$. The black hole is between the third and the first agent at a distance $x$ from one of them and at distance $y$ from the other (as in Figure~\ref{fig-impos-2-k}). By Lemma~\ref{lem:1}, at least one agent would fall into the black hole before any link to the black hole is identified. Consider the phase $P$ of the algorithm from the start until the first time an agent falls into the black hole. Let us call this agent $a$. Assume without loss of generality that agent $a$ enters the black hole by going right (i.e., after traveling a distance of $x$ from its homebase). First, notice that the agents never meet each-other during phase $P$ since they are always at a distance less than $2t$ from their homebases. 

Suppose for the sake of contradiction that after agent $a$ has vanished, the two surviving agents can identify the link used by agent $a$ to enter the black hole without sacrificing another agent. This is only possible if, whenever agent $a$ explores a new node to the right, it leaves a message encoding this fact %(If not, the adversary can set $x$ so that it  vanishes in the black hole, and the other agents cannot decide the exact position of the incident link of the black hole.)  Thus, for each new explored node, agent $a$ has to encode this fact 
and the only way to do this is by leaving another token. 
However, after $t$ explored nodes agent $a$ runs out of tokens. The adversary may then set $x$ to be any value between $t+1$ and $2t$. The remaining agents would not have enough information to determine the position of the black hole from the left (without the sacrifice of another agent). However, by Lemma~\ref{lem:1}, at least one agent must enter the black hole from the other link (on the right). Thus either one of the links to the black hole is never marked or there are no surviving agents.
\end{proof}

\begin{figure}[tb]
\begin{minipage}[t]{0.48\linewidth}
\begin{center}
\scalebox{0.8}{\includegraphics{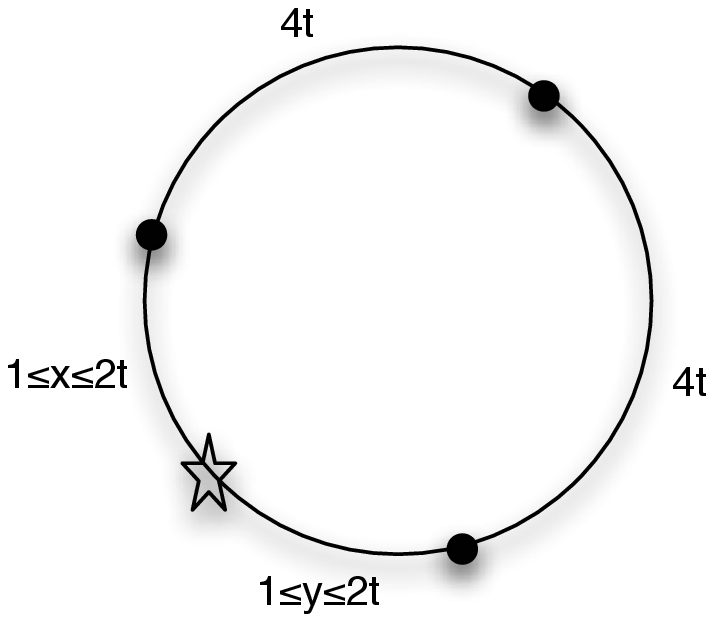}}
\end{center}
\vspace{-0.8cm}
\caption{\label{fig-impos-2-k}
Three agents with $t$ unmovable tokens each in an oriented ring.}
\end{minipage}
\hfill
\begin{minipage}[t]{0.48\linewidth}
\begin{center}
\scalebox{0.8}{\includegraphics{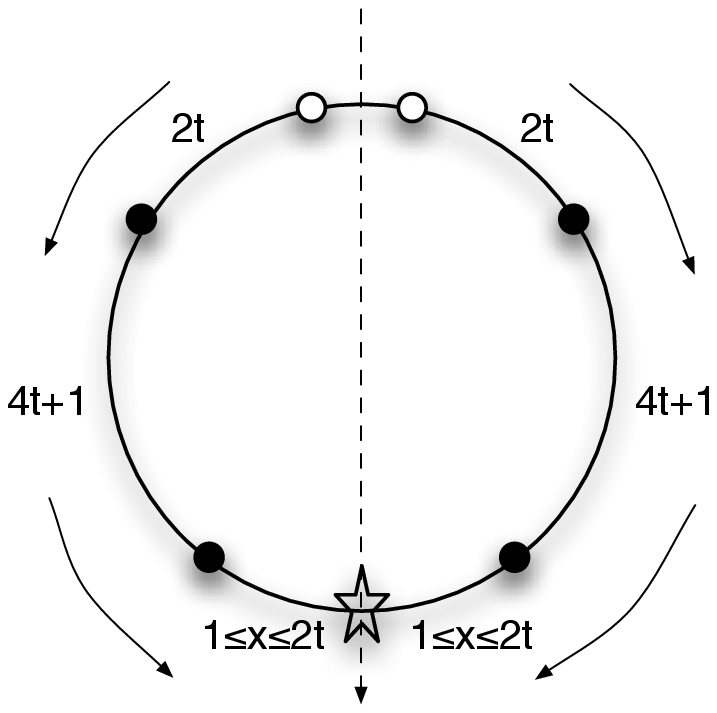}}
\end{center}
\vspace{-0.8cm}
\caption{\label{fig-impos-4-k}
Four agents with $t$ unmovable tokens each in an unoriented ring.}
\end{minipage}
\end{figure}

%-------------------------------------
\subsection{Unoriented Rings}
%-------------------------------------

In an unoriented ring, even four agents do not suffice to solve the BHS problem with unmovable tokens. In fact we show a stronger result that it is not even possible to identify just one of the links to the black hole, using four agents.

\begin{theorem}\label{th:impo-4unroriented}
In an unoriented ring, four agents carrying any constant number of unmovable tokens each, cannot correctly mark any link incident to the black hole, even when the agents have unlimited memory.
\end{theorem}

\begin{proof}
Suppose for the sake of contradiction that there exists an algorithm which marks one of the links incident to the black hole, using four agents each carrying $t$ unmovable tokens.  For some integer $x$, $1\leq x\leq 2t$, chosen by the adversary, suppose that the four agents and the black hole are initially placed as in Figure \ref{fig-impos-4-k}. The distance between two consecutive agents is 
$4t+1$ and the distance between the black hole and the closest agent on each side is $x$. Thus, the initial configuration is symmetric and the axis of symmetry crosses an edge and the black hole. The adversary can choose the orientations of the agents in such a way that the two agents closest to the black hole would fall into the black hole at the same time and before any agent meets another agent. Thus, the two surviving agents would continue to be in symmetric situation and they would take similar actions.
Using the same argument as in proof of Theorem \ref{th:impo-3oriented}, the information left by the vanished agent is not sufficient for one agent to correctly identify any link incident to the black hole. Due to the symmetry of the resulting configuration (and the fact that agents cannot meet on an edge), the two remaining agents can never meet and will always be in the same state until they both fall into the black hole (without marking any of the links incident to the black hole).
%By Lemma~\ref{lem:1}, another agent need to die to locate one of the link incident to the black hole. Because of the symmetry, the two remaining agents will die at the same step without marking any of the links.
\end{proof}

%-------------------------------------
\section{BHS Scheme with Movable Tokens}
%-------------------------------------

We first consider the case when the agents have movable tokens. If each agent has a movable token it can perform a cautious walk~\cite{dfps07}.
The \code{Cautious-Walk} procedure consists of the following actions: Put the token at the current node, move one step in the specified direction, return to pick up the token, and again move one step in the specified direction (carrying the token). After each invocation of the Cautious Walk, the agent looks at the configuration of the current node\footnote{Recall that only the tokens put on the node are counted, not the tokens carried by the agent itself.} and decides whether to continue performing Cautious Walk.

We show that only three agents are sufficient to solve BHS, when they have one movable token each. Algorithm~\ref{algo:BHS-R-3-1m} achieves this, both for oriented and unoriented rings. 
The procedure \code{Mark-Link} permanently marks as dangerous the specified link. %(to isolate the black hole from the network).

%****************************************************
\begin{algorithm}
\tcc{ BHS in any ring using $k\geq 3$ agents having $1$ movable token each} 
%Assumptions: Agents do not know k nor n. Tokens are identical. Agents do not communicate.} 
%-------------------------------------------------------------------
%\SetKwFunction{WAIT}{Wait}
\SetKwFunction{MarkLink}{Mark-Link}
\SetKwFunction{CautiousWalk}{CautiousWalk}
\SetKw{OR}{\textbf{or}}

\BlankLine
\lRepeat{current node has a token and no agent \OR next link is marked Dangerous}{\CautiousWalk(Left)\;}\;
\MarkLink{Left}\; 
\lRepeat{current node has a token and no agent \OR next link is marked Dangerous}{\CautiousWalk(Right)\;}\;
\MarkLink{Right}\; 

 \caption{BHS-Ring-1}
  \label{algo:BHS-R-3-1m}
\end{algorithm}
% END OF ALGORITHM
%**********************************************************

\begin{theorem}\label{lem:BHS-R-3-1m}
Algorithm \ref{algo:BHS-R-3-1m} solves the BHS problem in an unoriented ring with $k\geq 3$ agents having constant memory and one movable token each.
\end{theorem}

\begin{proof}
Notice that all agents start at the same time executing Procedure {\em CautiousWalk(dir)} and at each time they are at the same phase of this procedure. Since the only time an agent checks the number of tokens it sees is after completing an execution of the procedure, if the agent sees at least one token (not including the one it carries) then, either (i) there is another agent (i.e., they were traveling in opposite directions) or (ii) there is no other agent (which means that the token was left by an agent that disappeared). %We emphasize that checking the number of tokens or agents they see and marking the edges takes no time.
In case (i) the agent continues executing Procedure {\em CautiousWalk(dir)}. In case (ii) it is clear that the black hole resides at the next node in direction {\em dir}. In this case, the agent marks the edge to the black hole, reverses direction and repeats the process. Since the agents start from distinct locations at the same time, exactly two agents would fall into the black hole and both links to the black hole would be eventually marked.
%At the end (i.e., after the discovery of the second edge leading to the black hole), reverses again direction and moves one step. Now the other surviving agents will discover as well the black hole and stop.
\end{proof}

%-------------------------------------
\section{BHS Scheme with Unmovable Tokens}
%-------------------------------------

For agents having only unmovable tokens, we use the technique of \emph{Paired Walk} (called \emph{Probing} in \cite{ckr06}) for exploring new nodes. The procedure is executed by two co-located agents with different roles and the same orientation. One of the agents called the {\em leader} explores an unknown edge while the other agent, called {\em follower} waits for the leader. If the other endpoint of the edge is safe, the leader immediately returns to the previous node to inform the follower and then both move to this new node. On the other hand, if the leader does not return in two time steps, the follower knows that the next node is the black hole. 
(See Procedure \emph{Paired Walk}).% in the appendix.)

In order to use the \Pairedwalk technique, we need to gather two agents at the same node and then break the symmetry between them, so that distinct roles can be assigned to each of them. 
%Note that the latter task is difficult in an unoriented ring, if the two agents are in the same state when they meet. 
The basic idea of our algorithms is the following. We first identify the two homebases that are closest to the black hole (one on each side). These homebases are called \emph{gates}. The gates divide the ring into two segments: one segment contains the black hole (thus, is dangerous); the other segment contains all other homebases (and is safe). Initially all agents are in the safe part and an agent can move to the dangerous part only when it passes through the gate node. We ensure that any agent reaching a gate node, waits for a partner agent in order to perform the \Pairedwalk procedure.  
We now present two BHS algorithms, one for oriented rings and the other for unoriented rings.

%*********Procedure Paired Cautious Walk -------------------
\begin{procedure}
\SetKwFunction{WAIT}{Wait}

{ %\scriptsize
\uIf {isLeader} { 
Move one step in specified direction\; Move one step back\; Move one step in specified direction\;
}
\Else {
WAIT(2)\; \If{there is a leader} {Move one step in specified direction\; } 
}
}
\caption{PairedWalk( isLeader )}
  \label{proc:pairedCW}
\end{procedure}
%-----------------------------------------------------------------------

%****************************************************
\begin{algorithm}
\tcc{BHS in an Oriented Ring, using $k\geq 4$ agents having $2$ unmovable tokens each \\
Assumptions: All agents have the same initial state START.} 
%-------------------------------------------------------------------
% Agent's state: 
% GOBACK means an agent going home to put its second token
% ALONE means an agent without partner going to the single token node
% WAITING means an agent without partner that has reached a single token node (and must wait for a partner/follower, to start the joint-cautious-walk)
% LEADER and FOLLOWER : pair of agents performing joint-cautious-walk of the ring, if LEADER dies FOLLOWER marks link and become RIGHT-LEADER 
% RIGHT-LEADER and RIGHT-FOLLOWER : pair of agents performing joint-cautious-walk in the right direction, if RIGHT-LEADER dies RIGHT-FOLLOWER marks the right link
% LEFT (or RIGHT) SEARCHER : Move Left (or Right) to go to the gate and become RIGHT-FOLLOWER (i.e. you know that the danger on the left side has already been marked by someone)

%\SetAlgoNoEnd
%\LinesNotNumbered
%\SetVlineSkip{0.1cm}
%\SetKwFunction{WAIT}{Wait}
\SetKwFunction{MarkLink}{Mark-Link}
\SetKwFunction{PairedCW}{PairedWalk}
\SetKw{OR}{\textbf{or}} 
\SetKwBlock{Begin}{}{}

{ %\scriptsize

\BlankLine
START: \Begin{ Place token at homebase; $State:=$CHECK-LEFT\; }

CHECK-LEFT: \Begin{ Move Left until the next node that contains a token; $State:=$GO-BACK \; }

GO-BACK: \Begin(\tcc*[f]{Go home to put the second token}) { 
Move Right until the next node that contains a token; Put the second token\; 
\lIf{there is a LEADER agent} {$State:=$LEFT-SEARCHER\; }
\lElse { \lIf{there is a ALONE agent or a WAITING agent} {$State:=$FOLLOWER\; }  }
\lElse{ $State:=$ALONE\; }
}

ALONE: \Begin (\tcc*[f]{Move alone to the node with single token}) { Move Left until a node that contains either only one token \OR two tokens and a GO-BACK agent \;
\lIf { there is a RIGHT-LEADER agent} {$State:=$RIGHT-FOLLOWER\; }
\lElse{ \lIf{there is a GO-BACK agent and no WAITING agent} {$State:=$LEADER\; }}
\lElse{ \lIf{there is a WAITING agent and no GO-BACK agent} {$State:=$FOLLOWER\; }}
\lElse{ \lIf {there are no other agents} { $State:=$WAITING\; }}
}

WAITING: \Begin (\tcc*[f]{Wait for a partner agent}) { 
Wait until other agents arrive at the current node\; 
\lIf{there is a LEADER agent} {$State:=$LEFT-SEARCHER\; }
\lElse{ \lIf{there is a GO-BACK agent or ALONE agent} {$State:=$LEADER\; }}
\lElse{ \lIf { there is a RIGHT-LEADER agent} {$State:=$RIGHT-FOLLOWER\; } }
}

LEADER: \Begin (\tcc*[f]{Perform Paired walk with Follower agent}) {     
\lWhile{the left link is not marked dangerous} { \PairedCW(1) in Left direction\; }  
\lIf{the left link is marked dangerous} {$State:=$RIGHT-SEARCHER\; }
} 

FOLLOWER: \Begin (\tcc*[f]{Perform Paired walk with Leader agent}) {  
\While{there is a LEADER agent and the left link is not marked BH-link} { \PairedCW(0) in Left direction\; 
\If{the LEADER did not return during the last step} {\MarkLink(Left); $State:=$RIGHT-LEADER; exit loop\; }
}
%\lIf{the left link is marked BH-link} { $State:=$HALT\; }
}

LEFT-SEARCHER: \Begin (\tcc*[f]{Move Left to become a Right-Follower}) {  
 \lWhile{not at a node with one token} {Move Left\;}
 Wait until there is a RIGHT-LEADER agent; $State:=$RIGHT-FOLLOWER\; 
}

RIGHT-SEARCHER: \Begin (\tcc*[f]{Move Right to become a Right-Follower}) {  \lWhile{not at a node with one token} {Move Right\;}  
\lIf{there is a RIGHT-LEADER agent} {$State:=$RIGHT-FOLLOWER\; } 
%\lElse{ $State:=$HALT\; } 
}

RIGHT-LEADER: \Begin (\tcc*[f]{Perform Paired walk in the other direction}) {  
\lWhile{not at a node with one token} {Move Right\;}  
Wait until there is another agent;  %$State:=$RIGHT-LEADER\;  
\lWhile{true} { \PairedCW(1) in Right direction\; } 
} 

RIGHT-FOLLOWER: \Begin (\tcc*[f]{Perform Paired walk in the other direction}) {  
\While{there is a RIGHT-LEADER agent} {  \PairedCW(0) in Right direction\; 
\lIf{the RIGHT-LEADER did not return during the last step} {\MarkLink(Right); $State:=$HALT\; } 
}
}
%HALT: \Begin{Terminate the algorithm\;}
}

 \caption{BHS-Ring-2}
  \label{algo:BHS-R-4-2u}
\end{algorithm}
% END OF ALGORITHM
%**********************************************************

\bigskip

%****************************************************
% Algorithm for UnOriented Rings
\begin{algorithm}
\tcc{BHS in Unoriented Ring, using $k\geq 5$ agents having $2$ unmovable tokens each} 
%Assumptions: An agent can read the state of any agent that is on the same node. All agents have the same initial state START.} 
%-------------------------------------------------------------------
% 

%\SetAlgoNoEnd
%\LinesNotNumbered
%\SetVlineSkip{0.1cm}
%\SetKwFunction{WAIT}{Wait}
\SetKwFunction{MarkLink}{Mark-Link}
\SetKwFunction{PairedCW}{PairedWalk}
\SetKw{OR}{\textbf{or}} 
\SetKw{AND}{\textbf{and}} 
\SetKwBlock{Begin}{}{}

{ %\scriptsize
\BlankLine
START: \Begin{ Place token at homebase; $State:=$CHECK-LEFT\; }

CHECK-LEFT: \Begin{ Move Left until the next node that contains a token; $State:=$CHECK-RIGHT \; }

CHECK-RIGHT: \Begin{ Move Right until the next node that contains a token; 
Move Right again until the next node that contains some token; 
$State:=$GO-BACK \; }

GO-BACK: \Begin(\tcc*[f]{Go home to put the second token}) { 
Move Left until the next node that contains one token; Put the second token\; 
\lIf { there is a RIGHT-LEADER agent} {$State:=$RIGHT-FOLLOWER\; }
\lElse{ \lIf{there is a LEADER agent} {$State:=$SEARCHER\; } }
\lElse{ \lIf{there is a ALONE agent or a WAITING agent} {$State:=$FOLLOWER\; } } 
\lElse{ $State:=$ALONE\; }
}

ALONE: \Begin (\tcc*[f]{Move alone to the node with single token}) { Move Left until a node that contains either only one token \OR two tokens and a GO-BACK agent \;
\lIf{there are no other agents} {$State:=$WAITING\; }
\lIf { there is a RIGHT-LEADER agent} {$State:=$RIGHT-FOLLOWER\; }
\lElse{ \lIf{there is a LEADER agent} {$State:=$SEARCHER\; } }
\lElse{ \If{there is a WAITING agent \AND no GO-BACK agent} {
\eIf{ there is another ALONE agent having same orientation as the WAITING agent} {$State:=$SEARCHER\;} {$State:=$FOLLOWER\; }
}}
\lElse{ \If{there is a GO-BACK agent \AND no WAITING agent} {
\eIf{ there is another ALONE agent having same orientation as the GO-BACK agent} {$State:=$SEARCHER\;}{$State:=$LEADER\; }
}}
\lElse{ $State:=$SEARCHER\;  }
}

WAITING: \Begin (\tcc*[f]{Wait for a partner agent}) {  
Wait until other agents arrive at the current node\; 
\lIf{ there is a RIGHT-LEADER agent} {$State:=$RIGHT-FOLLOWER\; }
\lElse{ \If{there is a GO-BACK or ALONE agent $a$ \AND no LEADER agent} {
\eIf{ there is another WAITING agent having same orientation as agent $a$} {$State:=$SEARCHER\;}{$State:=$LEADER\; } 
}}
\lElse {$State:=$SEARCHER\;}
}

}

 \caption{BHS-Ring-3}
  \label{algo:BHS-R-5-2u}
\end{algorithm}
% END OF FIRST PART
%------------------------------------------------------
% START OF SECOND PART
\begin{algorithm}
\SetKwFunction{MarkLink}{Mark-Link}
\SetKwFunction{PairedCW}{PairedWalk}
\SetKwFunction{WAIT}{Wait}
\SetKw{OR}{\textbf{or}} 
\SetKw{AND}{\textbf{and}} 
\SetKwBlock{Begin}{}{}
%\nocaptionofalgo

{ %\scriptsize
\BlankLine
LEADER: \Begin{   
\lWhile{the left link is not marked dangerous} { \PairedCW(1) in Left direction\; }  
\lIf{the left link is marked dangerous} {$State:=$SEARCHER\; }
} 

FOLLOWER: \Begin{ 
Align orientation with the LEADER agent\;
\While{the left link is not marked BH-link} { 
	\PairedCW(0) in Left direction\; 
	\If{the LEADER did not return during the last step} { 
	     \MarkLink(Left); $State:=$RIGHT-LEADER; Exit Loop; 
	  }
%\lIf{the left link is marked BH-link} { $State:=$HALT\; }
}
}

SEARCHER: \Begin (\tcc*[f]{Go to the gate node to become RIGHT-FOLLOWER}){ 
\lWhile{not at a node with one token} {Move Right;} $State:=$RIGHT-FOLLOWER\; 
}

RIGHT-LEADER: \Begin(\tcc*[f]{Find a Follower and perform Paired Walk}){ 
\lWhile{not at a node with one token} {Move Right\;}  
\If {there is no other agent} {
\lRepeat{the current node contains one token}{Move Right; \WAIT(2)\;}}
\If {there is no other agent} {
\lRepeat{the current node contains one token}{Move Left\;}}
\lWhile{true} { \PairedCW(1) in Right direction\; } 
} 

RIGHT-FOLLOWER: \Begin{ 
Wait until there is a RIGHT-LEADER agent;
Align orientation with the RIGHT-LEADER agent\;
\While{there is a RIGHT-LEADER agent} {  
    \PairedCW(0) in Right direction\; 
    \If{the RIGHT-LEADER did not return during the last step} { 
        \MarkLink(Right); Exit loop\; } 
}
}
%HALT: \Begin{Terminate the algorithm\;}

} % End \tiny

 \caption{Algorithm BHS-Ring-3 (Continued)}
  %\label{algo:BHS-R-5-2u}
\end{algorithm}
% END OF ALGORITHM
%**********************************************************

%-----------------------------------
\subsection{Oriented Rings}

%Recall from the results of Section~\ref{} that we need at least $4$ agents with two unmovable tokens each to solve BHS in oriented rings. 
In an oriented ring, all agents may move in the same direction (i.e., Left). During the first phase of the algorithm each agent places a token on its homebase, moves left until the next homebase (i.e., next node with a token) and then returns to its homebase to put down the second token. During this phase one agent will fall into the black hole and there will be a unique homebase with a single token (a ``gate" node) and the other homebases will have two tokens each. However, the agents may not complete this phase of the algorithm at the same time. Thus during the algorithm, there may be multiple homebases that contain a single token. Whenever an agent reaches a ``single token" node, it waits for a partner and then performs \Pairedwalk in the left direction. One of the agents of a pair (the leader) eventually falls into the black hole and the other agent (the follower) marks the edge leading to the black hole and returns to the gate node, waiting for another partner. When another agent arrives at this node, these two agents perform \Pairedwalk in the opposite direction to find the other incident link to the black hole. The algorithm sketched below ensures that exactly one leader agent falls into the black hole from each side while performing \emph{Paired Walk}. 
(For the full algorithm, please see Algorithm~\ref{algo:BHS-R-4-2u}) %in the appendix.)

%The main idea of the algorithm is to break the symmetry of the ring using the black hole.
%Exactly one agent will die during the first phase of the algorithm, leaving only one token 
%before its death. This will permit the three other agents to meet on the only node with exactly 
%one token and form a team to find the black hole.
\bigskip 

\noindent \textbf{Algorithm BHS-Ring-2:} \\
\noindent During the algorithm, an agent $a$ performs the following actions.
\begin{enumerate} 
\item Agent $a$ puts a token and moves left until the next node with a token (state CHECK-LEFT) and then returns to its homebase $v$ (state GO-BACK) and puts its second token.
\item If there are no other agents at $v$, the agent moves left until it reaches a node containing exactly one token (state ALONE) and then waits for other agents arriving at this node (state WAITING).
 \item Otherwise, if there is a WAITING (or ALONE) agent $b$ at node $v$, the agents $a$ and $b$ form a (LEADER, FOLLOWER) pair.
\item If an ALONE agent meets a WAITING agent (and there are no other agents), they form a (LEADER, FOLLOWER) pair.
\item A LEADER agent performs \Pairedwalk until it falls into the black hole or it sees a link marked dangerous. In the latter case it moves to the gate node (state SEARCHER) and participates in \Pairedwalk in the other direction (state RIGHT-FOLLOWER).
\item A FOLLOWER agent performs  \Pairedwalk until the corresponding leader falls into the black hole or they see a link marked dangerous. In the former case, the agent (state RIGHT-LEADER) moves to the gate node and waits for a partner to start \Pairedwalk in the other direction.
\item When a WAITING agent $a$ meets a RIGHT-LEADER, agent $a$ becomes a RIGHT-FOLLOWER and participates in the \emph{Paired Walk}. 
\item The algorithm has some additional rules to ensure that no two LEADERs are created at the same node at the same time. No agent becomes a LEADER if there is already another LEADER at the same node (In this case, the agent become a SEARCHER and eventually a RIGHT-FOLLOWER when it reaches the gate node).
\end{enumerate}
\noindent When the algorithm BHS-Ring-2 is executed by four or more agents starting from distinct locations, %it is easy to see that 
the following properties hold: 
\begin{itemize}
\item Exactly one CHECK-LEFT agent falls into the black hole.
\item There is at least one LEADER agent and each LEADER has exactly one FOLLOWER.
\item No two LEADER agents are created at the same time on the same node and thus, two LEADERs can not reach the black hole at the same time.
\item There is exactly one RIGHT-LEADER agent and it falls into the black hole through the edge on the left side of the black hole.
\item An agent in any state other than CHECK-LEFT, LEADER, or RIGHT-LEADER, never enters the black hole.
\end{itemize}
%\end{lemma}

\begin{theorem}\label{th:BHS-R-4-2u}
Algorithm BHS-Ring-2 correctly solves the black hole search problem in any oriented ring with $4$ or more agents having constant memory and carrying two unmovable tokens each.
\end{theorem}

\begin{proof}(Sketch) :
Due to the above properties we know that at most three agents fall into the black hole. Thus there is at least one surviving agent at the end of the algorithm. Since there is at least one $(LEADER, FOLLOWER)$ pair moving in the left direction, one of the links to the black hole is discovered and marked by this pair of agents. The FOLLOWER agent of this pair becomes the RIGHT-LEADER. At this point, only two agents have been destroyed (one LEADER and one CHECK-LEFT agent). Thus there exists at least one other agent and this agent will eventually reach the gate node, i.e., the only node having a single token (Note that an agent may wait only at node containing a single token). Thus, this agent will become the RIGHT-FOLLOWER and perform a \Pairedwalk with the RIGHT-LEADER to discover the other link to the black hole.
\end{proof}

 \subsection{Unoriented Rings} 
 
For unoriented rings, we need at least $5$ agents with two unmovable tokens each. The algorithm for unoriented rings with unmovable tokens is similar to the one for oriented rings, except that each agent chooses an orientation. When two agents meet and one has to follow the other, we assume that the state of the agent contains information about the orientation of the agent (i.e., the port at the current node considered by the agent to be \emph{Left}). Thus, when two agents meet at a node, one agent (e.g. the Follower) can orient itself according to the direction of the other agent (e.g. the Leader). 
\bigskip 

\noindent \textbf{Algorithm BHS-Ring-3:} \\
\noindent Each agent puts one token on its homebase, goes on \emph{its left} until it sees another token and then returns to its homebase. Now the agent goes on its right until it sees a token and then returns again to the homebase. The agent now puts its second token on its homebase. During this operation exactly two agents will fall into the black hole. Each surviving agent walks to its left until it sees a node $u$ with a single token. At this point the agent has to wait, since either there is a black hole ahead, or $u$ is the homebase of an agent $b$ that has not returned yet to put its second token.

It may happen that two agents arrive at node $u$ at the same time from opposite directions. In this case, both agents can wait until another agent arrives. Note that in this case, the ring is safe on both directions until the next homebase and thus, an agent $b$ (whose homebase is $u$) would arrive within a finite time. When agent $b$ arrives, only one of the waiting agents (the one having the same orientation as $b$) changes to state LEADER and pairs-up with agent $b$.
A similar case occurs when an agent $a$ is waiting and two agents (both ALONE) arrive from different directions. Among these two agents, the one having the same orientation as agent $a$ pairs up with agent $a$ and starts the \Pairedwalk procedure.

As before there can be multiple leader-follower pairs performing \Pairedwalk in different parts of the ring. Note that no two LEADERs can be created at the same node at the same time. Thus, two LEADERs may not enter into the black hole at the same time from the same direction. After the first LEADER enters the black hole from one direction, the corresponding FOLLOWER agent marks the link as a dangerous link and thus, no other agents enter the black hole from the same direction.

We ensure that each LEADER agent has exactly one FOLLOWER agent. When the LEADER agent falls into the black hole, the corresponding FOLLOWER agent becomes the RIGHT-LEADER. The objective of the RIGHT-LEADER is to discover the other link incident to the black hole. The RIGHT-LEADER agent moves to the other end of the ring until the node with one token. Since we assume there are at least five agents, there must be either an unpaired agent at one of the gates or, there must be another (LEADER, FOLLOWER) pair that has already detected and marked the other link leading to the black hole. If the RIGHT-LEADER does not find a RIGHT-FOLLOWER at the first gate, it performs a \emph{slow} walk to the other gate and returns again to the former gate.
During the slow walk, it moves at one-third the speed of any other agent (i.e., waits two steps after each move). This ensures that it will meet another agent in at least one of the two gates. These two
agents now starts the \Pairedwalk procedure in the other direction.

A complete description of the algorithm can be found in Algorithm~\ref{algo:BHS-R-5-2u}. %in the appendix. 
The following properties can be verified:
%\begin{lemma} During the algorithm BHS-Ring-3, the following holds, assuming there are at least $5$ agents, each carrying two unmovable tokens
\begin{enumerate}
\item Exactly two agents fall into the black hole before placing their second token.
\item There is at least one LEADER and each LEADER has a corresponding FOLLOWER.
\item There is either one or two RIGHT-LEADER agents (with opposite orientations).
\item At most one LEADER or RIGHT-LEADER enters the black hole from each direction.
\item An agent in any other state never enters the black hole after placing its second token.
\end{enumerate}
%\end{lemma}

Due to the above properties, we know that at most $4$ agents may fall into the black hole. We now show that both links to the black hole are actually discovered and marked as dangerous, during the algorithm.

\begin{theorem}\label{th:BHS-R-5-2u}
Algorithm~BHS-Ring-3 correctly solves the black hole search problem in unoriented ring with 5 or more agents having constant memory and carrying two unmovable tokens each.
\end{theorem}

\begin{proof}(Sketch)
Let $v$ be the black hole and $u$ and $w$ be the gates. Assume nodes $u$, $v$ and $w$ appear in this order (say, in clockwise direction) with only non-homebase nodes between $u$ and $v$ and between $v$ and $w$. Let us denote by $d(u,v)$, $d(v,w)$ and $d(w,u)$ the distances between these nodes in the clockwise direction.
Due to Property~2 above we know that there is at least one LEADER. Suppose there are two LEADERs moving in opposite directions. In this case, the two LEADER agents will fall into the black hole through distinct links, while executing the \Pairedwalk procedure. Thus, both links to the black hole will be marked by the corresponding FOLLOWER agents and BHS is solved.
 
The second scenario to consider is when there are multiple LEADER agents but they all have the same orientation. In this case, no two LEADERs are created at the same node at the same time. Only the first LEADER to reach the black hole falls into it and thus there will be exactly one RIGHT-LEADER agent (the FOLLOWER of the agent that falls into a black hole becomes RIGHT-LEADER). All the other LEADERs become RIGHT-FOLLOWER and they go to one of the gate nodes (say $u$) and wait for the RIGHT-LEADER. We need to show that the RIGHT-LEADER would meet at least one RIGHT-FOLLOWER. There is a time difference of at most $d(w,u)$ between the RIGHT-LEADER and any RIGHT-FOLLOWER. The RIGHT-LEADER traverses a path of length $2(d(u,v)-1)+2d(w,u)$ before reaching the node $u$ for the second time. During this time any RIGHT-FOLLOWER agent would have certainly reached node $u$ and would be waiting for the RIGHT-LEADER. Thus, the RIGHT-LEADER meets at least one RIGHT-FOLLOWER and they perform the \Pairedwalk in the other direction to discover the other link to the black hole.

The third and only other possible scenario is when there is only one LEADER agent (and thus only one FOLLOWER). We know that one of the links to the black hole would be discovered by this (LEADER, FOLLOWER) pair. The FOLLOWER agent would eventually become RIGHT-LEADER. We need to show that the RIGHT-LEADER would meet at least one RIGHT-FOLLOWER.
Suppose (w.l.o.g.) that the RIGHT-LEADER was created on the same side of the black hole as node $u$. 
Since there are at least $5$ agents, there must be an unpaired agent waiting for a partner. The WAITING agent could be waiting either at node $u$ or at node $w$. Note that the RIGHT-LEADER moves from $u$ to $w$, waiting two steps after each edge traversal. This takes time $3d(u,w)$. Since no other agent makes more than $3d(u,w)$ edge traversals before it starts waiting at the gate, the RIGHT-LEADER is guaranteed to find at least one waiting agent either at node $w$ or at $u$ when it returns. Thus this agent becomes RIGHT-FOLLOWER and joins the \Pairedwalk procedure. Hence, the second link to the black hole would be discovered by this pair of agents.
%
%In the first case, it is easy to show using a simple counting argument that when the RIGHT-LEADER reach node $u$ for the second time (i.e., just before starting the \Pairedwalk) the other agent is already waiting at node $u$. Thus they start the \Pairedwalk together. In the second scenario, the WAITING agent is waiting at the node $w$ and in this case, the RIGHT-LEADER may have missed it when it enters node $w$ for the first time. However, once the RIGHT-LEADER start the \Pairedwalk procedure, it will eventually arrive at node $w$ before going to the black hole. At this point, the other agent is guaranteed to be at node $w$ and thus this agent becomes RIGHT-FOLLOWER and joins the \Pairedwalk procedure. Hence, the second link to the black hole would be discovered by this pair of agents.      
\end{proof}

%----------------------------
\section{Conclusions}
%---------------------------
The results of this paper determine the minimum resources necessary for locating a black hole in synchronous ring networks. We presented algorithms that use the optimal number of agents and the optimal number of tokens per agent, while requiring only constant-size memory. Thus, all resources used by our algorithms are independent of the size of the network. Notice that all the algorithms presented in the paper have a time complexity of $O(n)$ steps, so, they are asymptotically optimal for BHS in a ring. 
The main question answered by the paper is how the limitation on the memory of the agents influences the resources required for solving BHS. We show that the constant memory limitation has no influence on the resource requirements since the (matching) lower bounds hold even if the agents have unlimited memory.
It would be interesting to investigate if similar tight results hold for BHS in other network topologies. We would also like to investigate the difference between `pure' and `enhanced' token model in terms of the minimum resources necessary for black hole search in higher degree networks.

%Summarize the results and comment on the importance (constant memory, only pure tokens, tightness of results etc.) 
%Write about the possibility of generalizing to other topologies (as a hint to our next paper on the torus ;))
%NOTE (for the journal version): If the agents are allowed an unlimited supply of tokens, the results on the minimum number of agents do not hold anymore.

%\newpage
%\pagenumbering{roman}

\bibliographystyle{abbrv}
%\bibliography{/Users/emarkou/Documents/RESEARCH/bibliography/rendezvous/biblio-tot}
\bibliography{biblio-tot}

\end{document}